\documentclass[11pt, letterpaper]{article}
\title{Adversarial Low Degree Testing}
\author{
Dor Minzer\thanks{Department of Mathematics, Massachusetts Institute of Technology, Cambridge, USA. Supported by a Sloan Research
Fellowship, NSF CCF award 2227876 and 
NSF CAREER award 2239160.}
\and
Kai Zhe Zheng\thanks{Department of Mathematics, Massachusetts Institute of Technology, Cambridge, USA. Supported by the NSF GRFP.}
}
\date{\vspace{-5ex}}
\usepackage{fullpage}
\usepackage{amsmath,amssymb,amsfonts,nicefrac}
\usepackage{amsthm}
\usepackage{xspace}
\usepackage{color}
\usepackage{url}
\usepackage{hyperref}
\usepackage{enumitem}
\usepackage{bm}
\usepackage{bbm}
\usepackage{times}
\usepackage{mathtools}
\usepackage{enumitem}
\usepackage{algorithm}
\usepackage{subfig}
\usepackage{algpseudocode}
\usepackage{graphicx}
\DeclareMathOperator{\rej}{rej}

\DeclareMathOperator{\poly}{poly}
\DeclareMathOperator{\spa}{span}

\DeclareMathOperator{\RM}{RM}

\newcommand{\Ff}{\mathbb{F}}

\newcommand{\T}{\mathcal{T}}

\newcommand{\eps}{\varepsilon}

\renewcommand{\epsilon}{\eps}

\newcommand\inner[2]{\langle{#1},{#2}\rangle}

\renewcommand\leq{\leqslant}
\renewcommand\geq{\geqslant}

\theoremstyle{plain}
\newtheorem{theorem}{Theorem}[section]
   \newtheorem{thm}{Theorem}[section]
   
   \newtheorem{lemma}[thm]{Lemma}
   
   \newtheorem{claim}[thm]{Claim}
   
   \newtheorem{definition}{Definition}

\DeclareMathOperator\supp{supp}

\begin{document}
\maketitle

\begin{abstract}
    In the $t$-online-erasure model in property testing, an adversary is allowed to erase $t$ values of a queried function for each query the tester makes. This
    model was recently formulated by Kalemaj, Raskhodnikova and 
    Varma, who showed that the properties of linearity of functions 
    as well as quadraticity can be tested in 
    $O_t(1)$ many queries: $O(\log (t))$ for linearity and $2^{2^{O(t)}}$
    for quadraticity. They asked whether the more general 
    property of low-degreeness can be tested in the online erasure model, whether better testers exist for quadraticity, and if similar results hold when ``erasures'' are replaced with ``corruptions''.

    We show that, in the $t$-online-erasure model,
    for a prime power $q$, given query access to a function 
    $f: \mathbb{F}_q^n \xrightarrow[]{} \mathbb{F}_q$, one 
    can distinguish in $\poly(\log^{d+q}(t)/\delta)$ queries between the case that $f$ is degree at most $d$, and the case that $f$ is $\delta$-far from any degree $d$ function (with respect to the fractional hamming distance). 
    This answers the aforementioned questions and brings the query complexity to nearly match the query complexity of low-degree testing in the classical property testing model. 
    
    Our results are based on 
    the observation that the
    property of low-degreeness
    admits a large and 
    versatile family of
    query efficient testers. 
    Our testers operates by 
    querying a uniformly random, sufficiently 
    large set of points in a large enough affine subspace, and finding 
    a tester for low-degreeness that only
    utilizes queries from 
    that set of points.
    We believe that this tester may find other applications to algorithms in the online-erasure model or other related models, and may be of independent interest.
\end{abstract}

\section{Introduction}
The main purpose of this
paper is to further investigate the 
$t$-online-erasure model 
which was recently initiated by Kalemaj, Raskhodnikova, and Varma~\cite{KRV}. The motivation behind this model stems from scenarios where computations are performed on datasets that may contain erasures which can 
 occur in places that are chosen by an adversary but are limited in their number. %Although these erasures may not necessarily be introduced adversarially, analyzing the adversarial model enables us to provide worst-case guarantees. Moreover, there exist naturally occurring adversarial instances. 
For instance, consider an algorithm aiming to detect fraud, while an adversarial party deliberately erases data to conceal the fraudulent activity.

The focus of this paper is the problem of \emph{low degree testing} in the $t$-online erasure model. In this setting, an algorithm
is given query access to a function $f: \Ff_q^n \xrightarrow[]{} \Ff_q$, 
an erasure parameter $t$, a distance parameter $\delta$, and a degree parameter $d$. At the start of the computation, all entries in $f$'s truth table are unerased. An algorithm may query the value of $f(x)$ for $x \in \Ff_q^n$, but after each query, an adversary is allowed to \emph{erase} any $t$-entries of $f$'s truth table. If the algorithm attempts to query an entry $f(x)$ after it is erased, then $\perp$ is returned. The goal of the algorithm is to output accept with probability $1$ if $f$ is degree at most $d$ and to output reject with probability at least $2/3$ if $f$ is $\delta$-far from low degree. In the standard property testing literature, the former stipulation is often referred to as completeness, while the latter is often referred to as soundness. Algorithms accomplishing this task in the $t$-online erasure model are called \emph{$t$-online-erasure-resilient}.

Low degree testing is one of the most basic and well studied problems in property testing and can be traced all the way back to Blum, Luby, and Rubinfeld's seminal work on linearity testing \cite{BLR}. Linearity testing can be thought of as the $d = 1$ case of low degree testing, and after the work of \cite{BLR}, further  low degree tests were constructed and analyzed over the next decade - first for arbitrary degrees of univariate polynomials \cite{RSu}, and later for arbitrary degrees of multivariate polynomials \cite{AKKLR, KR, JPRZ, BKSSZ, HSS, RZS, KM, MZ}. These tests were also relevant within coding theory as polynomials (both univariate and multivariate) over finite fields are the basis of two of the most well known error-correcting codes - the Reed-Solomon and Reed-Muller Codes. Due to their efficient query complexity, low degree tests for multivariate polynomials also found many applications elsewhere in theoretical computer science - perhaps most notably to the construction Probabilistically Checkable Proofs \cite{FGLSS,  AroraSafra, ALMSS, RubinfeldSudan, AroraSudan}. The extensive work regarding low degree testing also led to the testing of other properties of functions over finite fields - an area often called algebraic property testing, \cite{KS}. 

Our work adds to the vast literature on low degree testing by giving the first such algorithms over all degrees and field sizes in the online erasure model. Previously, the only known results were due to \cite{KRV} who gave algorithms for linearity and quadraticity ($d=1,2$) cases over the field $\Ff_2$. 

We find that designing a tester in the online-erasure model requires one to depart significantly from classical testers. To see why, consider what happens when one tries to implement the Blum-Luby-Rubinfeld linearity tester of \cite{BLR} over $\Ff_2$. The BLR tester is based on the observation that any linear $f: \Ff_2^n \xrightarrow[]{} \Ff_2$ satisfies $f(x+y) = f(x) + f(y)$ for any $x, y \in \Ff_2^n$ and proceeds by choosing uniformly random $x, y \in \Ff_2^n$, querying $f(x), f(y), f(x+y)$, and finally checking if $f(x+y) = f(x) + f(y)$ is satisfied. The difficulties that arise in the online erasure model now become clear. After the algorithm queries $f(x)$ and $f(y)$, the adversary can erase the point $f(x+y)$ making it impossible to directly carry out the BLR test.

Faced with this difficulty, 
the authors of~\cite{KRV}
noted that one can attempt to ``overwhelm'' the adversary's erasure capacity with volume. One can first query a large number, say $m = 100t$, of uniformly random points, $x_1,\ldots, x_m$ so that there are $\binom{m}{2} \geq 4000t^2$ possible $x_i + x_j$'s. At this point the adversary could have erased at most $100t^2$ of these sums, making it possible to successfully query a random $x_i + x_j$, and perform the check $f(x_i+x_j) = f(x_i) + f(x_j)$. This idea is adapted and refined into a tester for linearity over $\Ff_2$ by \cite{KRV}.

Dealing with online erasures becomes significantly more complicated for higher degrees, namely for $d\geq 2$. The classical test here proceeds by choosing a $d+1$-dimensional affine space $V$ and then querying $f(x)$ for any $x \in V$, \cite{AKKLR}. Given the large number of points needed as well as the large number of dependencies between these points, it is no longer clear how to craft enough strategies to obtain an entire $d+1$-dimensional affine subspace against an adversary. Indeed, the 
tester of~\cite{KRV} in
the online erasure 
model is much more involved, 
and it is not clear how
to generalize it to 
larger fields or higher degrees.

\paragraph{Our approach.}
Our algorithm forgoes the task of mimicking
the classical low-degree 
testers by attempting
to query all of the points in a subspace of $\mathbb{F}_q^n$ altogether.
Instead, we show that given any suitable number of points inside of a fixed affine subspace of dimension large enough in terms of $d$ and $t$, one has enough points to design a degree $d$ tester. Thus, our algorithm circumvents the adversary by querying uniformly random points inside of some suitably large affine subspace. As it turns out, the adversary cannot foil our plan if there is no plan. 

To the best of our knowledge, our low degree tester using uniformly random points is new, even in the setting without erasures, and we believe that it may have other applications and be of independent interest.

\subsection{Our Results}
% We design local tests for the property of being a low degree polynomial. This problem is parameterized by an input function $f: \Ff_q^n \xrightarrow[]{} \Ff_q$, a degree parameter $d$, a distance parameter $\delta$. As we work in the $t$-online-erasure model, there is an additional erasure budget $t$. It will be helpful to think of $q$ as large, $d$ as a large multiple of $q$, $n$ as going to infinity relative to $d$, and $\delta$ as a small constant fraction. Our testers will hold for $t$ up to roughly $O(\delta q^{n/d})$ and our work addresses three of the open questions from \cite{KRV}.

We present our testers for low degree polynomials over prime fields and non-prime fields separately. Although the testers are essentially the same, the analysis in the non-prime field case requires more care and must be handled separately. This also leads to slightly different parameters for our two testers. Formally, our main theorems are the following.

\begin{theorem} [Prime Field Case] \label{th: main p}
    Let $f: \Ff_p^n \xrightarrow[]{} \Ff_p$ be the input function over a prime-field vector space, $d$ be the degree parameter, and $\delta$ be the distance parameter. Then for $t \leq \frac{\delta}{30}p^{n/(20d)}$, there is a $t$-online-erasure resilient tester with query complexity    $O\left(\frac{\left(\log^{3d+3}(t/\delta)\right)}{\delta} \right)$
    satisfying:
    \begin{itemize}
        \item \textbf{Completeness:} If $f$ is degree $d$ then the algorithm outputs accept with probability $1$.
        \item \textbf{Soundness:} If $f$ is $\delta$-far from degree $d$ then the algorithm outputs reject with probability at least $2/3$.
    \end{itemize}
\end{theorem}

\begin{theorem} [Non-Prime Field Case] \label{th: main q}
    Let $f: \Ff_q^n \xrightarrow[]{} \Ff_q$ be the input function, $d$ be the degree parameter, and $\delta$ be the distance parameter. Then for $t \leq \frac{\delta}{100}q^{\frac{n}{100(d+q)} - 1}$, there is a $t$-online-erasure resilient tester with query complexity $\frac{q^{O(1)}}{\delta} O\left(\left(\log^{3d+3q}(t/\delta)\right)\right)$
    satisfying:
    \begin{itemize}
        \item \textbf{Completeness:} If $f$ is degree $d$ then the algorithm outputs accept with probability $1$.
        \item \textbf{Soundness:} If $f$ is $\delta$-far from degree $d$ then the algorithm outputs reject with probability at least $2/3$.
    \end{itemize}
\end{theorem}
In the $p=2$ and $d = 1$ case, our algorithm is essentially the same as that of \cite{KRV}. Furthermore, we remark that our testers also hold with two sided error in the $t$-online corruption model with two sided error. In this model, the adversary may change the value of points $f(x)$ instead of erasing them. Our algorithms are still effective in the online-corruption because no matter what points the adversary chooses to erase, the probability that our algorithm queries a corrupted point is bounded by a small constant. Therefore the corruptions only induce a small additive error to the completeness and soundness cases. Our result for corruptions follows from a reduction from~\cite{KRV}, showing that online erasure-resilience implies online corruption-resilience if the probability of querying a corrupted point is small (see~\cite[Lemma 1.8]{KRV} 
for a precise statement).

\begin{theorem}
    The algorithms of Theorems~\ref{th: main p} and ~\ref{th: main q} are also $t$-online-corruption resilient and satisfy:
    \begin{itemize}
        \item \textbf{Completeness:} If $f$ is degree $d$ then the algorithm outputs accept with probability $2/3$.
        \item \textbf{Soundness:} If $f$ is $\delta$-far from degree $d$ then the algorithm outputs reject with probability at least $2/3$.
    \end{itemize}
\end{theorem}
In the online corruption model, one can never hope for perfect completeness as there is always a nonzero probability that a large enough number of queries made are corrupted, and there is no way for the algorithm to tell if a query has been corrupted or not. In contrast, in the erasure model, the algorithm can always tell if a query is erased as $\perp$ is returned.

\subsection{Proof Overview} The starting point of our algorithm is the observation that the degree of a function when restricted to a $k$-dimensional affine subspace can be tested using \emph{any} set of points of size only polynomial in $k$. This result is obtained by considering inner product testers and reviewing properties of linear affine-invariant functions \cite{KS}, which are explained in Section~\ref{sec: preliminaries}. Specifically, we rely on the fact that a degree $d$ tester for functions over prime fields is sound if there exists even a single degree $d+1$ monomial that it can reject. Thus our task now becomes designing a tester that accepts all degree at most $d$ monomials while rejecting at least one degree $d+1$ monomial. The view of testing as taking the inner product with functions that are dual to the desired property, as in is done in \cite{RZS, MZ}, turns out to be useful for this task. This further reduces our algorithm to the following algebraic observation: for any $S \subset \Ff_q^k$ of size roughly $k^d$, there is a function $h: \Ff_q^k \xrightarrow[]{} \Ff_q$ whose support is contained within $S$ such that the inner product of $h$ with any monomial of degree at most $d$ is $0$, and the inner product of $h$ with some specific degree $d+1$ monomial is nonzero. 
This fact follows essentially follows by counting arguments/ dimension arguments, and armed with 
it the algorithm follows rather naturally.

The case of non prime fields 
is a bit harder but is 
very similar in spirit. Here, 
we once again show that testing
low-degreeness over $\mathbb{F}_q$ admits a versatile and large enough
family of inner product testers, enough so that 
we can fit a test through every large enough set of 
points.

\subsection{Related Work} \label{sec: related work}
Prior to this work, the $t$-online erasure model was presented and studied in \cite{KRV}. In \cite{KRV}, it is shown that the query complexity of linearity testing over $\Ff_2$ with erasure parameter $t$ is $\Theta(\log(t))$, and that quadraticity testing can be accomplished using a number of queries that is doubly exponential in $t$. As such, the quadraticity testing result only applies for constant $t$ and furthermore does not apply in the online corruption model. Other properties including the sorted and lipschitz properties of sequences are also considered in \cite{KRV} as examples that cannot be tested in the online-erasure model, even for $t=1$.

The online erasure model we study is also similar to the offline erasure model of Dixit, Raskhodnikova, Thakurta, and Varma \cite{DRTV} where all erasures are made by the adversary at the start of the computation, as well as the tolerant testing model, where there is the additional requirement that functions close to the desired property should be accepted with constant probability \cite{PRR}. 

In a concurrent work, Ben-Eliezer, Kelman, Meir and Raskhodnikova~\cite{BKMR} investigate the online erasure model in the $q = 2$ case. Of interest to this paper are the following results. For linearity testing ($d = 1$), they obtain a tester with query complexity matching the lower bound in \cite{KRV}. For general degree, $d$, they show that any tester requires at least  $\Omega\left( \log^d(t) \right)$ queries. 

For constant distance parameter $\delta$, in the $q=2$ case, our tester in Theorem~\ref{th: main p} achieves query complexity $O(\log^{3d+3}(t))$. Thus, when compared to the lower bound of ~\cite{BKMR} the polynomial dependence in $\log^d(t)$ is optimal, but improving the exponent to match the lower bound is an intriguing open question. We remark that we did not attempt to optimize the query complexity of our testers and instead focused on obtaining a tester with $\poly(\log^d(t))$ query complexity that also worked for all degrees and field sizes. Optimizing just the analysis in this paper, however, will not close the gap to the lower bound. Instead, what appears to be missing is an ``optimal analysis'' of the soundness of our basic tester (see Algorithm~\ref{alg: basic tester p}), that makes $O(\log^{d+1}(t))$ queries. In this paper, we rely on a generic result from \cite{KS} to obtain soundness roughly $\Omega(\frac{1}{\log^{2d+2}(t)})$, which in turn requires us to repeat the basic test $O(\log^{2d+2}(t))$-times in order to obtain constant soundness. This ultimately results in $O(\log^{3d+3}(t))$ queries. If instead, we had an optimal analysis, similar to what is achieved for the testers in \cite{MZ}, then the soundness of the basic test would be $\Omega(1)$, and only $O(1)$ repetitions would be needed. Unfortunately the analysis of \cite{MZ} does not apply to our testers.

\section{Preliminaries} \label{sec: preliminaries}
\subsection{Notations}
For a prime power $q = p^{\ell}$, let $\Ff_q$ denote the finite field with $q$ elements and $\Ff_q^{*}$ denote the nonzero elements of $\Ff_q$. When $\ell = 1$, then $q = p$ and we refer to $\Ff_q = \Ff_p$ as a prime field. For a fixed field, $\Ff_q$, let $\T_{n, k}$ denote the set of affine transformations $T: \Ff_q^{n} \xrightarrow[]{} \Ff_q^{k}$. Though the definition of $\T_{n,k}$ also depends on a field parameter $q$, we will often drop it from the notation and it will always be clear from the context. For $T \in \mathcal{T}_{n,k}$ let $f \circ T$ denote the function over $\Ff_q^k$ given by $f \circ T(x) := f(T(x))$, for $x \in \Ff_q^k$.

For $\alpha \in \Ff_q^n$, and $e \in \{0,\ldots, q-1\}^n$, let $\alpha^e = \prod_{i=1}^n \alpha_i^{e_i}$, where $\alpha_i$ and $e_i$ are the $i$th coordinates of $\alpha$ and $e$ respectively. Let $x^e$ denote the monomial $\prod_{i=1}^{n}x_i^{e_i}$ where the $x_i$'s are variables and the number of variables parameter $n$ will be clear from context. The degree of $x^e$ is $\deg(x^e) = |e|_1 = \sum_{i=1}^n e_i$. For a function $f: \Ff_q^n \xrightarrow[]{} \Ff_q$, recall that it has a unique polynomial expression, $f(x) = \sum_{e \in \{0,\ldots, q-1\}^n} C_e x^{e}$, where each coefficient $C_e \in \Ff_q$. The degree of $f$ is the maximum degree of a monomial appearing in its expansion, or equivalently $\deg(f) = \max_{e: C_e \neq 0} |e|_1 $. For each $C_e \neq 0$, we say that the monomial $x^e$ \textit{appears} in $f$ or that $f$ \emph{contains} $x^e$. 

As mentioned, the set of degree $d$ polynomials is often referred to as the degree $d$ Reed-Muller Code. We use the notation $\RM[n,q,d] = \{f: \Ff_q^n \xrightarrow[]{} \Ff_q \; | \; \deg(f) \leq d \}$. The notion of distance that we use is fractional hamming distance. That is, for two functions $f,g$ over $\Ff_q^n$, define $\delta(f,g) = \Pr_{x \in \Ff_q^n}[f(x) \neq g(x)]$.
The distance of $f$ to degree $d$, is $\delta_d(f) = \min_{g \in \RM[n,q,d]} \delta(f,g)$,
and we say that $f$ is $\delta$-far from degree $d$ if $\delta_d(f) \geq \delta$.

\subsection{Properties of Linear Affine Invariant Families}
We require a few basic facts about linear, affine invariant families of functions, which we present in this section. These results can all be found or derived from \cite{KS}. For a family of polynomials $\mathcal{F}$, let \[
\supp({\mathcal{F}}) = \{x^e \; | \; e \in \{0,\ldots, q-1\}^n, \; \exists f \in \mathcal{F} \text{ containing } x^e  \}\] 
denote the set of monomials that appear in at least one of these polynomials. It is well known that these monomials form a basis for $\mathcal{F}$ if $\mathcal{F}$ is \textit{linear} and \textit{affine-invariant}. These notions are defined as follows.

\begin{definition}
    A family of functions, $\mathcal{F} \subset \{f: \Ff_q^n \xrightarrow[]{} \Ff_q\}$ is linear if 
    \begin{itemize}
        \item $0 \in \mathcal{F}$, where $0$ is the constant zero function.
        \item For any $\alpha, \beta \in \Ff_q$, if $f, g \in \mathcal{F}$, then $\alpha \cdot f + \beta \cdot g \in \mathcal{F}$.
    \end{itemize}
\end{definition}

\begin{definition}
    A family of functions, $\mathcal{F} \subset \{f: \Ff_q^n \xrightarrow[]{} \Ff_q\}$, is called affine invariant if it is closed under compositions with affine transformations. That is, if $f \in \mathcal{F}$ and $T \in \T_{n,n}$, then $f \circ T \in \mathcal{F}$.
\end{definition}

Linear, affine-invariant function families have the following nice properties. 
Towards this end, given a family of function 
$\mathcal{F}$, we denote by 
$\supp(\mathcal{F})$ the set of monomials that appear in some function in $\mathcal{F}$. 
The following property asserts that affine invariant codes are defined by a set of monomials .

\begin{lemma} [Monomial Extraction Lemma \cite{KS}] \label{lm: monomial extraction}
    If $\mathcal{F}$ is a linear, affine-invariant family of polynomials then $\mathcal{F} = \spa(\supp(\mathcal{F}))$.
\end{lemma}

The next property is called $p$-shadow closed, and requires some additional notions. For two integers $a, b \in \{0, \ldots, q-1 \}$, let $a = \sum_{i=0}^{k-1}p^ia_i$ and $b = \sum_{i=0}^{k-1}p^ib_i$ be their base $p$ representations. We say $a$ is in the $p$-shadow of $b$ if $a_i \leq b_i$ for $i = 0, \ldots, k-1$, and denote this by $a \leq_p b$. Then for two exponent vectors $e = (e_1, \ldots, e_n)$ and $e' = (e'_1, \ldots, e'_n)$, we say $e \leq_p e'$ if $e_i \leq_p e'_i$ for every $i$. Linear, affine-invariant families of polynomials have the following shadow closed property.

\begin{lemma} \label{lm: shadow closed}
    Let $\mathcal{F} = \spa(\supp(\mathcal{F}))$ be a linear, affine invariant family of polynomials. If $e \leq_p e'$ and $e' \in \supp(\mathcal{F})$, then $e \in \supp(\mathcal{F})$ as well.
\end{lemma}

In the case where $q = p$, the shadow lemma simplifies to the following statement.
\begin{lemma} \label{lm: shadow closed p}
    Let $\mathcal{F} = \spa(\supp(\mathcal{F}))$ be a linear, affine invariant family of polynomials over a prime field vector space. If $x^e \in \supp(\mathcal{F})$ and each entry of $e - e'$ is nonnegative, then $x^{e'} \in \supp(\mathcal{F})$.
\end{lemma}

Finally, we will need the following lemma which will allow us to go from one monomial in $\mathcal{F}$ to another with the same degree, but with the distribution of the individual degrees shifted.

\begin{lemma} \label{lm: monomial shift}
    Suppose $x^e \in \mathcal{F}$ and suppose $m \leq_p e_2$. Then $x^{e'}\in \mathcal{F}$, where $e' = (e_1 + m, e_2-m, e_3, \ldots, e_n)$.
\end{lemma}
\begin{proof}
    Let $T$ be the affine transformation $T(x) = (x_1, x_1+x_2, x_3, \ldots, x_n)$. Then, 
    \[
    x^e \circ T = x_1^{e_1}(x_1+x_2)^{e_2} \prod_{j=3}^n x_j^{e_j} = \left(\sum_{i=0}^{d_2} \binom{d_2}{i}x_1^{e_1+i}x_2^{e_2-i} \right)  \prod_{j=3}^n x_j^{e_j}.
    \]
    By Lucas's Theorem and the assumption $m \leq_p e_2$, $\binom{d_2}{m} \neq 0$ in $\Ff_q$, and so $x^e \circ T = x_1^{e_1}(x_1+x_2)^{e_2} \prod_{j=3}^n x_j^{e_j} $ contains the monomial $x^{e'}$. The result then follows from Lemma~\ref{lm: monomial extraction}.
\end{proof}
Clearly Lemma~\ref{lm: monomial shift} also holds with the indices $1$ and $2$ replaced by arbitrary, distinct indices $i$ and $j$.
\subsection{Functions over Finite Fields}
We now present some additional facts about functions over finite fields. These facts will lead to the inner-product view of testing which was used to construct low-degree tests in \cite{RZS, MZ}. For two functions $f, g: \Ff_q^{n} \xrightarrow[]{} \Ff_q$, define their inner product as
$$\langle f, g \rangle = \sum_{\alpha \in \Ff_q^n} f(\alpha)g(\alpha).$$
It is clear that this inner product is bi-linear. A key part of our testing algorithm is the observation that the presence of high degree monomials can be ``detected'' by using inner products with polynomials of relatively sparse support. In order to see this, we will first need a couple of basic facts about finite fields. First, it is 
a well known fact that $\mathbb{F}_q^{*}$ has a multiplicative generator. From this, we can deduce the following two lemmas.

\begin{lemma} \label{lm: power sums}
For any prime power $q$ and integer $i \in  \{0,\ldots, q-1\}$,
\begin{equation*}
    \sum_{\alpha \in \Ff_q}\alpha^i =
    \begin{cases}
      -1, & \text{if}\ i = q-1, \\
      0, & \text{otherwise}.
    \end{cases}
  \end{equation*}
\end{lemma}

\begin{proof}
If $i = q-1$, then $\alpha^i = 1$ for all $\alpha \neq 0$, while $0^i = 0$. Therefore, the sum is one summed up $q-1$ times which is $-1$ in $\Ff_q$. For $i \in \{1, \ldots, q-2\}$, recall that $\Ff_q^*$ has a generator $\gamma$. That is, $\Ff_q^* = \{1, \gamma, \ldots, \gamma^{q-2}\}$. Since $\gamma^i \neq 1$, we may write
\[
\sum_{\alpha \in \Ff_q} \alpha^i = \sum_{j = 0}^{q-2} (\gamma^i)^j = %\gamma^i \cdot 
\frac{(\gamma^i)^{q-1}-1}{\gamma^i - 1} = 
%\gamma^i \cdot 
\frac{1-1}{\gamma^i-1} = 0.
\]

On the other hand, if $i = 0$ then the sum on the left hand side of the lemma is equal to $\sum_{\alpha \in \Ff_q} \alpha^0 = \sum_{\alpha \in \Ff_q} 1 = q = 0$.
\end{proof}

\begin{lemma} \label{lm: monomial inner product}
    For $e,e'\in  \{0,\ldots, q-1\}^n$, suppose that  $x^{e''}= x^{e}\cdot x^{e'}$ where $x^{e''}$ is reduced so that its individual degrees are in $\{0, \ldots, q-1\}$.
    \begin{equation*}
    \langle x^{e}, x^{e'} \rangle =
    \begin{cases}
      (-1)^n, & \text{if}\ e'' = (q-1, \ldots, q-1), \\
      0, & \text{otherwise},
    \end{cases}
  \end{equation*}
\end{lemma}
\begin{proof}
    This is a straightforward application of Lemma~\ref{lm: power sums}. By definition
    \[
     \langle x^{e}, x^{e'} \rangle = \sum_{(\alpha_1, \ldots, \alpha_n) \in \Ff_q^n} \prod_{i=1}^n \alpha_i^{e_i} = \prod_{i=1}^n \sum_{\alpha \in \Ff_q} \alpha^{e''_i}.
    \]
    The result follows from Lemma~\ref{lm: power sums}.
\end{proof}

Lemma~\ref{lm: monomial inner product} suggest one method for using inner products to test low degree. Suppose that $q= 2$ and we want to detect if $f$ contains the degree $d+1$ monomial $x_1 \cdots x_{d+1}$. Then by Lemma~\ref{lm: monomial inner product} this is equivalent to checking whether or not $\langle f, x_{d+2} \cdots x_n \rangle $ is nonzero.

\subsection{Local Characterizations and Testing}
Using the ideas from the last section, we can first attempt to find functions $h: \Ff_q^k \xrightarrow[]{} \Ff_q$ that are \emph{local characterizations} of $\RM[n,q,d]$. These are defined as follows.

\begin{definition} \label{def: local char}
For $h: \Ff_q^k \xrightarrow[]{} \Ff_q$, define $\mathcal{F}_n(h) = \{f: \Ff_q^n \xrightarrow[]{} \Ff_q \; | \; \langle f \circ T, h \rangle = 0, \forall T \in \T_{n,k} \}$. We call $h$ a local characterization for $\mathcal{F}_n(h)$.
\end{definition}
Whenever we use the notation $\mathcal{F}_n(h)$, the field size $q$ will be clear from context. Once we are able to construct a local characterization, $h$, a candidate test becomes clear. Choose a random $T \in \T_{n,k}$, and accept if and only if $\langle f \circ T, h \rangle = 0$. Carrying out this test requires querying $f \circ T(x)$ for every $x \in \supp(h)$. It is clear from the definition of local characterization that this test satisfies completeness, while for soundness we may appeal to a general result from \cite{KS}. For our purposes, this lemma states the following.

\begin{theorem}[Lemma 2.9 \cite{KS}]\label{th: local char to test}
     For any local characterization $h: \Ff_q^k \xrightarrow[]{} \Ff_q$ and function $f: \Ff_q^n \xrightarrow[]{} \Ff_q$, we have
    \[
    \Pr_{T \in \T_{n,k}}[\inner{f \circ T}{h} \neq 0 ] \geq  \frac{\delta}{4Q^2},
    \]
    where $Q = \supp(h)$ and $\delta = \delta(f, \mathcal{F}_n(h))$.
\end{theorem}

\section{The Prime Field Case: Proof of Theorem~\ref{th: main p}}

Fix a degree parameter $d$, a prime field size $q = p$, and an erasure parameter $t$. In this section we prove Theorem~\ref{th: main p} by  describing and analyzing a $t$-erasure-resilient degree $d$ tester for functions over $\Ff_p^n$. For distance parameter $\delta$, our test is $t$-online-erasure resilient for all $t \leq \frac{\delta}{30}p^{n/(20d)}$.

\subsection{A Local Characterization for  \texorpdfstring{$\RM[n,p,d]$}{RM[n,p,d]}}
To start, we give a sufficient condition for a function $h$ to be a local characterization $\RM[n,p,d]$. This requires the following Corollary of Lemma~\ref{lm: monomial shift} applied to the prime field case.
\begin{lemma} \label{lm: p shadow closed}
    Let $\mathcal{F}$ be a linear affine invariant family of functions over $\Ff_p^n$ and suppose $x^e \in \mathcal{F}$ for some $e \in \{0,\ldots, p-1\}^n$. Then $x^{e'} \in \mathcal{F}$ for any $e' \in  \{0,\ldots, p-1\}^n$ such that $|e'|_1 \leq |e|_1$,
\end{lemma}
\begin{proof}
    Note that for $a, b \in \{0,\ldots, p-1\}$ the relation $a \leq_p b$ is equivalent to $a \leq b$. Thus, by repeatedly apply Lemma~\ref{lm: monomial shift}, it is clear that $x^{v} \in \mathcal{F}$ for every $v \in \{0,\ldots, p-1\}^n$ such that $|v|_1 = |e|_1$. In particular, there exists a $v$ such that $e' \leq_p v$ and $x^v \in \mathcal{F}$. By Lemma~\ref{lm: shadow closed}, it follows that $x^{e'} \in \mathcal{F}$.
\end{proof}

Broadly, Lemma~\ref{lm: p shadow closed} states that in order for $h$ to be a local characterization for $\RM[n,p,d]$, it is sufficient for $h$ to detect any monomial of degree $d+1$. This idea is formalized in the lemma below.

\begin{lemma} \label{lm: sufficient for local char}
For any $k \geq  \lceil \frac{d+1}{p-1} \rceil$, if $h: \Ff_p^k \xrightarrow[]{} \Ff_2$ satisfies $\deg(h) = k(p-1) - (d+1)$ then, $h$ is a local characterization for $\RM[n,p,d]$. Equivalently $\RM[n,p,d] = \mathcal{F}_n(h)$.
\end{lemma}

\begin{proof}
    Let $h: \Ff_p^k \xrightarrow[]{} \Ff_p$ have degree $k(p-1) - (d+1)$. We show that $\RM[n,p,d] = \mathcal{F}_n(h)$.

    By Lemma~\ref{lm: monomial inner product}, $\RM[n,p,d] \subseteq \mathcal{F}_n(h)$. Indeed, for any $f \in \RM[n,p,d]$ and any affine transformation $T \in \T_{n,k}$, we have the degree of $f \circ T$ is at most $d$. Thus if $x^e$ is a monomial appearing in $f \circ T$ and $x^{e'}$ is a monomial appearing in $h$, then $|e+e'|_1 \leq k(p-1) - (d+1) + d \leq k(p-1) - 1$. It follows that $e + e' \neq (q-1, \ldots, q-1)$, and by Lemma~\ref{lm: monomial inner product}, $\langle  x^e, x^{e'} \rangle = 0$. Since this holds for every pair of monomials appearing in $f \circ T$ and $h$, by the bilinearity of the inner product, $\langle f \circ T, h \rangle = 0$ for every $T \in \T_{n,k}$.

    To complete the proof we show that $\mathcal{F}_n(h) \subseteq \RM[n,p,d]$. Suppose $\prod_{i=1}^{k} x_i^{e_i}$ is one of the degree $k(p-1) - (d+1)$ monomials in $h$. Let $e'_i = p-1-e_i$ and define 
    \[
    e' = (e'_1, \ldots, e'_k, 0, \ldots, 0) \in \{0,\ldots, p-1\}^n.
    \]
    Let $T$ be the affine transformation that is identity on the first $k$-coordinates and sends all other coordinates to $0$, i.e.\ $T(x_1, \ldots, x_n) = (x_1, \ldots, x_k, 0, \ldots, 0)$. By Lemma~\ref{lm: monomial inner product}, $\langle x^{e'} \circ T, h \rangle \neq 0$ it follows that $ x^{e'} \notin \mathcal{F}_n(h)$. Since $|e'|_1 = d+1$ Lemma~\ref{lm: p shadow closed} it follows that $\mathcal{F}_n(h)$ contains no monomials of degree greater than $d$. Thus, $\mathcal{F}_n(h) \subseteq \RM[n,p,d]$.
\end{proof}

\subsection{An Erasure Resilient Low Degree Tester}
In order to design a low degree tester that works in the presence of online erasures, we would like to have that has as little dependency between the points as possible. To this end, we show that simply choosing $\approx k^d$ uniformly random points in a $k$-dimensional tester for $k \geq d+1$ to be determined later suffices for low degree testing. To make this tester $t$-online-erasure resilient, we pick $k$ sufficiently large relative to the erasure budget $t$.

\begin{lemma} \label{lm: dual poly}
For every dimension $k \geq d+1$ and every $S \subseteq \Ff_p^k$ such that $|S| \geq \binom{d+k+1}{k}+1$, there exist a function $h: \Ff_p^k \xrightarrow[]{} \Ff_p$ such that 
\begin{itemize}
    \item $\deg(h) = k(p-1) - (d+1)$
    \item $\supp(H) \subset |S|$
\end{itemize}
\end{lemma}
\begin{proof}
    Let $\mathcal{H}$ be the set of functions from $\Ff_p^k \xrightarrow[]{} \Ff_p$ whose support is contained in $S$. Then
    \[
    |\mathcal{H}| \geq p^{|S|} \geq p^{\binom{d+k+1}{k}+1}.
    \]
    On the other hand, the number of monomials over $\Ff_p^k$ of degree at least $k(p-1)-(d+1)$ is $\binom{d+k+1}{k}$. To see this, we can bound the number of $e \in \{0,\ldots, p-1\}^k$ such that $\sum_{i=1}^k p-1-e_i \leq d+1$.  The number of linear combinations of such monomials over $\Ff_p$ is at most $p^{\binom{d+k+1}{k}}$. By the pigeonhole principle there must exist distinct $h_1, h_2 \in \mathcal{H}$ that have the exact same coefficient for each monomial of degree at least $k(p-1)-(d+1)$. Set $h' = h_1 - h_2$, so that all of these monomials are cancelled out in $h'$. Then, $\deg(h') \leq k(p-1)-(d+1)$. If this inequality is in fact an equality, then we are done. Otherwise, let $x^e$ be a maximum degree monomial appearing in $h'$. Let $x^v$ be a monomial of minimum degree such that $x^e \cdot x^v$ has degree $k(p-1)-(d+1)$. Such a monomial indeed exists and we can take $h = h' \cdot x^v$ to get the desired function. Indeed, $x^e \cdot x^v$ is a maximum degree monomial of $h$ and $\supp(h) \subseteq \supp(h') \subseteq S$.
    % Now take $h_1$ and $h_2$ as described in the first paragraph. If either $h_1$ or $h_2$ has degree $k-(d+1)$, then we are done, as that would be the desired $h$. Otherwise, multiply $h_1$ by some monomial $m$ so that $\deg(h_1 m ) = k-(d+1)$ and set $h = h_1 m + h_2$. It is clear that $h$ satisfies the first two items. To see that it satisfies the third, notice that by design, $h_1$ and $h_2$ have disjoint supports so $\supp(h) \supset \supp(h_2)$.
\end{proof}

Combining Lemmas~\ref{lm: dual poly} and \ref{lm: sufficient for local char}, we get that on any set $S \subseteq \Ff_p^k$ of size $\binom{d+k+1}{k}+1$, there is a local characterization $h$ whose support is contained inside $S$. This motivates the basic tester described in Algorithm~\ref{alg: basic tester p}. We will analyze Algorithm~\ref{alg: basic tester p} ignoring erasures. Thus, we give a lower bound on its rejection when assuming that no points are erased. To deal with erasures, we note that each query is uniformly random in a $k$-dimensional affine subspace. Hence by choosing $k$ large enough later, we can give an upper bound on the probability that any query we make is erased. For the remainder of this section, fix $$Q(k) = \binom{d+k+1}{k}+1.$$ 

\begin{algorithm}[!ht]
\caption{A Basic Low Degree Tester over $\Ff_p$}
\label{alg: basic tester p}
\begin{algorithmic}[1]
\Procedure{RandomPoints}{$f, d, k$}
\State Choose $T \in \T_{n,k}$ uniformly at random.
\State Choose $S \subseteq \Ff_q^k$ of size $Q(k)$ uniformly at random.
\State Query $f \circ T(x)$ for each $x \in S$.
\State $\mathcal{H}(S) \gets $ the set of $h$ with support contained in $S$ that satisfy $\RM[n,p,d] = \mathcal{F}_n(h)$.
\If{ $\mathcal{H}(S)$ is empty}
\State \Return Accept.
\EndIf

\State Choose $h \in \mathcal{H}(S)$ uniformly at random. 
\If{ $\inner{f \circ T}{h} = 0$}
\State \Return Accept.
\EndIf
\If{ $\inner{f \circ T}{h} \neq 0$}
\State \Return Reject.
\EndIf
\EndProcedure
\end{algorithmic}
\end{algorithm}

\begin{lemma} \label{lm: random points tester}
    Let $f: \Ff_p^n \xrightarrow[]{} \Ff_p$ be a function. If $\deg(f) \leq d$, then Algorithm~\ref{alg: basic tester p} outputs accept with probability $1$. If $f$ is $\delta$-far from degree $d$, Algorithm~\ref{alg: basic tester p} outputs rejects with probability at least $\frac{\delta}{4Q(k)^2}$.
\end{lemma}
\begin{proof}
    We begin with the completeness of the tester, and towards 
    this end assume that $f$ has degree at most $d$.
    Since $\deg(h)\leq k(p-1) - (d+1)$ for all $h \in \mathcal{H}$, and $\deg(f \circ T) \leq \deg(f)$, it follows from Lemma~\ref{lm: monomial inner product} and the bilinearity of inner product that $\langle f \circ T, h \rangle = 0$ for all $T$. 
    %Indeed, expanding $(f \circ T) \cdot h$ can only yield monomials of degree up to $k(p-1) - 1$ and thus cannot yield the monomial $x_1^{p-1} \cdots x_k^{p-1}$.
 
    We move on to the second part of the lemma. For $h \in \mathcal{H}$, let $\Pr(h)$ denote the probability that $h$ is chosen in step $3$ and let $\rej_h(f)$ denote the probability over $T \in \T_{n,k}$ that $\langle f \circ T, h \rangle \neq 0$. Notice that the random $T$ and $h$ in Algorithm~\ref{alg: basic tester p} induce a product distribution over $\T_{n,k} \times \mathcal{H}$ where the distribution of $\T_{n,k}$ is uniform. Conditioned on any arbitrary $h \in \mathcal{H}$ being chosen in line 9, $T$ is uniform over $\T_{n,k}$ and the test rejects with probability at least $\rej_h(f)$. By Theorem~\ref{th: local char to test}, for every $h \in \mathcal{H}$, $\rej_h(f)\geq \frac{\delta}{4|\supp(h)|^2} \geq  \frac{\delta}{4Q(k)^2}$.

    Thus, by generating the distribution over $T$ and $h$ by first choosing $h \in \mathcal{H}$ with probability $\Pr(h)$ and then choosing $T \in \T_{n,k}$ uniformly. This makes it clear that Test 1 rejects with probability at least
    \[
    \sum_{h \in \mathcal{H}}\Pr(h) \rej_h(f) \geq  \frac{\delta}{4Q(k)^2}. \qedhere
    \]
\end{proof}

In order to output reject with probability at least $2/3$ in the far from degree $d$ case, we will repeat Algorithm~\ref{alg: basic tester p}. Henceforth, set $k = 20d\log_p(30t/\delta)$. We will show that the resulting tester, described in Algorithm~\ref{alg: fp test},
is $t$-online-erasure resilient. Before showing completeness and soundness, we give a bound on $Q(k)$ which will be helpful for calculations,
\[
Q(k) \leq (3k/d)^{d+1}  \leq (60 \log_p(30t/\delta))^{d+1}.
\]
By the bound on $t$ stated at the start of the section, we have $k \leq n$, so our tester is valid, in that $f \circ T$ is always well defined for $T \in \T_{n,k}$.

\begin{algorithm}[!ht]
\caption{A $t$-online-erasure resilient tester over $\Ff_p$.}
\label{alg: fp test}
\begin{algorithmic}[1]
\Procedure{ErasureResilient}{$f, d, \delta$}
\State Set $k = 100d\log \left( \frac{100dt}{\delta}\right) + q$.
\For{$i = 1 \to \frac{100Q(k)^2}{\delta}$} \State Run \texttt{RandomPointsTest}$(f,d,k)$
\State If Reject is outputted, \Return Reject.
\EndFor
\State \Return Accept.
\EndProcedure
\end{algorithmic}
\end{algorithm}

The following result shows completeness and soundness for Algorithm~\ref{alg: basic tester p} and hence proves Theorem~\ref{th: main p}
\begin{theorem}
Given $f: \Ff_p^n \xrightarrow[]{} \Ff_p$, a degree parameter $d$, a distance parameter $\delta$, and an erasure parameter $t$, Algorithm~\ref{alg: fp test} accepts $f$ with probability $1$ if $\deg(f) \leq d$ and rejects $f$ with probability at least $2/3$ if $f$ is $\delta$-far from degree $d$. The number of queries made is
$O\left(\left(\frac{\log(t/\delta)}{\delta}\right)^{3d+3} \right)$.
% \[
% \frac{100(60 \log_p(30t/\delta))^{3d+3}}{\delta} = O\left(\left(\frac{\log(t/\delta)}{\delta}\right)^{3d+3} \right)
% \]
\end{theorem}
\begin{proof}
    It is clear from Lemma~\ref{lm: random points tester} that if $\deg(f) \leq d$, then $f$ is accepted with probability $1$.

    Now suppose that $f$ is $\delta$-far from degree $d$. By Lemma~\ref{lm: random points tester}, with probability at least $3/4$, there is at least one iteration where the basic tester outputs reject, so it remains to bound the probability that any queried point is erased.
    The total number of erased points is at most $\frac{100Q(k)^3t}{\delta}$
    and each query is uniformly random in some $k$-dimensional affine subspace. Therefore by a union bound, the probability that any individual query is erased is at most,
    \[
     \frac{100Q(k)^3t}{\delta p^k} \cdot \frac{Q(k)^3t}{\delta} \leq \frac{100t^2(30 \log_p(60t/\delta))^{6d+6}}{\delta^2(30t/\delta)^{20d}} \leq \frac{1}{100}.
    \]
    Overall, this shows that if $f$ is $\delta$-far from degree $d$ then  Algorithm~\ref{alg: fq test} outputs rejects with probability at least $3/4 - 1/100 \geq 2/3$.
\end{proof}

\section{The Non-Prime Field Case}
Our algorithm for low degree testing over non-prime fields follows the same outline as that of the prime field case, but its analysis includes a few more complications. For the remainder of this section, write $d+1 = s \cdot (q-q/p) + r$ for an integer $s$ and an integer $0 \leq r < q-q/p$. For distance parameter $\delta$, our test is $t$-online-erasure resilient for all $t \leq \frac{\delta}{100}q^{n/(100(d+q)) - 1}$.

\subsection{Local Characterizations for  \texorpdfstring{$\RM[n,q,d]$}{RM[n,q,d]}}
Our first goal is to show the following lemma which gives sufficient conditions for $h: \Ff_q^k \xrightarrow[]{} \Ff_q$ to be a local characterization. Its proof has appeared before, first in \cite{RZS} and later restated closer to the language used here in the appendix of \cite{MZ}. We include it below for completeness.

\begin{lemma} \label{lm: sufficient local char q}
    If $h: \Ff_q^k \xrightarrow[]{} \Ff_q$ satisfies 
    \begin{itemize}
        \item $\deg(h) \leq (q-1)\cdot k - (d+1)$.
        \item $h$ contains the monomial $\left(\prod_{i=1}^{s}x_i^{q/p-1}\right)x_{s+1}^{q-1-r'} \left(\prod_{j =s+2}^{k} x_{j}^{q-1}\right)$ for all $r \leq r' \leq q-1$,
    \end{itemize}
    then $\RM[n,q,d] = \mathcal{F}_n(h)$.
\end{lemma}
We split the proof of Lemma~\ref{lm: sufficient local char q} into two claims. Fix $h: \Ff_q^k \xrightarrow[]{} \Ff_q$ to be a function satisfying the conditions above.
\begin{claim}
    $\RM[n,q,d] \subseteq \mathcal{F}_n(h)$.
\end{claim}
\begin{proof}
    This is a direct application of Lemma~\ref{lm: monomial inner product} and the same as the first part of the proof of Lemma~\ref{lm: sufficient for local char}
\end{proof} 
\begin{claim}
    $\mathcal{F}_n(h) \subseteq \RM[n,q,d]$.
\end{claim}
\begin{proof}
Suppose for the sake of contradiction that there is $g \in \mathcal{F}_n(h)$ with degree greater than $d$ and let $x^e$ be a maximum degree monomial of $g$. By Lemma~\ref{lm: monomial extraction}, $x^e \in \mathcal{F}_n(h)$. Take the smallest index $i$ such that $e_1 + \cdots + e_i > d$, and let $e' = (e_1, \ldots, e_i, 0, \ldots, 0) \in \{0,\ldots, q-1\}^k$. Then, $e' \leq_p e$ so $x^{e'} \in \mathcal{F}_n(h)$, and $|e'|_1 \geq s(q-q/p) + r'$ for some $r \leq r' \leq q-1$.

Next we will repeatedly use Lemma~\ref{lm: monomial shift} to obtain a monomial in $\mathcal{F}_n(h)$ where each of the first $s$ individual degrees are at least $q-q/p$. To this end, define
\[
   c(e) = \sum_{i=1}^{s} \max(0, (q-q/p)-e_i).
\]
We will abuse notation and still refer to the monomial as $x^{e'}$ after each application of Lemma~\ref{lm: monomial shift}. If $e'$ is not of the desired form, then one of the following must be true:
\begin{itemize}
       \item There is $j > s$ such that $e_j > 0$, in which case we simply find some $p^{m} \leq_p e_j$ and apply Lemma~\ref{lm: monomial shift} to obtain the monomial $x_i^{e_i+p^m}x_j^{e_j - p^m}$ in place of $x_i^{e_i}x_j^{e_j}$,
       \item There is $j \leq s$ such that $e_j > q-q/p$. In this case we can find $p^m$ such that $e_j - p^m \geq q-q/p$, and apply Lemma~\ref{lm: monomial shift} to obtain the monomial $x_i^{e_i+p^m}x_j^{e_j - p^m}$ in place of $x_i^{e_i}x_j^{e_j}$.
\end{itemize}
In either case, $c(e')$ strictly decreases, while $|e'|_1$ does not change, so we will eventually end up with $e'$ such that $e'_i \geq q-q/p$ for $1 \leq i \leq s$, and $|e'|_1 = s(q-q/p)+r'$. At this point $\sum_{i = 1}^{s}e'_i - (q-q/p) + \sum_{i = s+1}^{k} e'_i = r' \leq q-1$ so we may apply Lemma~\ref{lm: monomial shift} to shift all degree above $q-q/p$ in the first $s$ coordinates, and all degree in the $s+2$ through $k$th coordinates, onto $e'_i$, to obtain the monomial
\[
\left(\prod_{i=1}^s x_i^{q-q/p}\right) x_{s+1}^{r'},
\]
where $r \leq r' \leq q-1$. As we only applied Lemma~\ref{lm: monomial shift}, this monomial must be in $\mathcal{F}_n(h)$. However, by the second condition of Lemma~\ref{lm: sufficient local char q} and Lemma~\ref{lm: monomial inner product}, this monomial and $h$ have nonzero inner product, leading to a contradiction. It follows that $\mathcal{F}_n(h)$ does not contain any functions of degree greater than $d$, completing the proof.
\end{proof}

\subsection{An Erasure Resilient Low Degree Tester}
We now describe our $t$-erasure resilient degree $d$ tester. Similar to the prime field case, we show that a set of random points contains a local characterization, and then we repeat this tester to reject with constant probability in the far from degree $d$ case. 

Fix a dimension $k$ to be determined later. Set $e^{\star} \in \{0, \ldots, q-1\}^{k}$ such that
\begin{itemize}
    \item $e^{\star}_i= q-q/p$ for $1 \leq i \leq s$,
    \item $e^{\star}_{s+1} = q-1$,
    \item $e^{\star}_j = 0$ for $j \geq s+2$.
\end{itemize}
Let $d^{\star} = |e^{\star}|_1$, let $\mathcal{G} = \{g: \Ff_q^k \xrightarrow[]{} \Ff_q \; | \; \deg(g) \leq d^{\star}-1\}$, and let $\mathcal{E} = \{e \in \{0,\ldots, q-1\}^k \; | \; |e|_1 \leq d^{\star}-1 \}$. Note that $d^{\star} \leq d+q$, $s \leq \lceil \frac{d+1}{q-q/p}\rceil$, $|\mathcal{E}| \leq \binom{k+d^{\star}}{d^{\star}}$, and $|\mathcal{G}| \leq q^{\binom{k+d^{\star}}{d^{\star}}}$. Now construct the following system of equations over the variables $\{z_{\alpha}\}_{\alpha \in S}$:
\begin{itemize}
    \item For each $e' \in \mathcal{E}$, include the equation $\sum_{\alpha \in S} z_{\alpha} \alpha^{e'} = 0$.
    \item Include the equation $\sum_{\alpha \in S} z_{\alpha} \alpha^{e^{\star}} = 1$.
\end{itemize}

We first claim that over uniformly random $S \subseteq \Ff_q^k$ of size $Q(k) = 2q^{s+1}\log(q)\binom{k+d^{\star}}{d^{\star}}$, this system of equations is solvable with high probability. To show this, we first state the following simple observation, which can also be seen through the Schwartz-Zippel lemma or combinatorial nullstellensatz. 

\begin{lemma} \label{lm: Schwartz Zippel Modification}
    For any $g \in \mathcal{G}$,
    \[
    \Pr_{\alpha \in \Ff_q^k}[g(\alpha) \neq x^{e^{\star}}(\alpha)] \geq \frac{1}{q^{s+1}}.
    \]
\end{lemma}
\begin{proof}
    For any $(\alpha_{s+2}, \ldots, \alpha_k) \in \Ff_q^{k-s-1}$ the $s+1$-variate polynomial obtained by setting $x_i = \alpha_i$ for $s+2 \leq i \leq k$ into $g(x) - x^{e^{\star}}$ is nonzero as it contains the monomial $\prod_{i=1}^{s+1}x_i^{e^{\star}_i}$. Thus for any setting of the last $k-s-1$ variables, there exists $(\alpha_1, \ldots, \alpha_{s+1})$ such that $g(x) - x^{e^{\star}} \neq 0$ for $x = (\alpha_1, \ldots, \alpha_{s+1}, \alpha_{s+2}, \ldots, \alpha_k)$ and the lemma follows.
\end{proof}

\begin{lemma} \label{lm: solvable system}
    Choose $S \subseteq \Ff_q^k$ uniformly at random of size $Q(k)$. Then the probability that the above system of equations over $\{z_{\alpha}\}_{\alpha \in S}$ has no solution is at most $q^{-\binom{k+d^{\star}}{d^{\star}}}$.
\end{lemma}
\begin{proof}
The system will have a solution as long as the coefficients in the last equation are not in the span of the coefficients of the other equation. In other words, if
\[
    (\alpha^{e^{\star}})_{\alpha \in S} \notin \spa\{ (\alpha^{e'})_{\alpha\in S}\}_{e' \in \mathcal{E}}.
\]
However, the span on the right side is simply the set $\{(g(\alpha)_{\alpha \in S} \; | \; g \in \mathcal{G} \}$. Therefore the above event will happen if for every $g \in \mathcal{G}$ there exists at least one $\alpha \in S$ such that $g(\alpha) \neq \alpha^{e^{\star}}$. By Lemma~\ref{lm: Schwartz Zippel Modification}, for any individual $g$, we have
\[
\Pr_{S}[g(\alpha) = \alpha^{e^{\star}}, \forall \alpha \in S] \leq (1-q^{-(s+1)})^{Q(k)} \leq e^{-q^{-(s+1)}Q(k)}.
\]
Therefore, by a union bound, the probability that there is at least one $g \in \mathcal{G}$ equal to $x^{e^{\star}}$ on all of $S$ is at most
\[
|\mathcal{G}|e^{-q^{-(s+1)}Q(k)} \leq e^{-\log(q)\binom{k+d^{\star}}{d^{\star}}} \leq q^{-\binom{k+d^{\star}}{d^{\star}}}.
\qedhere
\]
\end{proof}

Finally, we note that if the system of equations we constructed is solvable, then there exists $h: \Ff_q^k \xrightarrow[]{} \Ff_q$ with support contained in $S$ that satisfies the conditions of Lemma~\ref{lm: sufficient local char q}, and is thus a local characterization of $\RM[n,q,d]$.

\begin{lemma}
    Choose $S \subseteq \Ff_q^k$ uniformly at random of size $Q(k)$. Then with probability at least $1 - q^{-\binom{k+d^{\star}}{d^{\star}}}$, there exists $h: \Ff_q^k \xrightarrow[]{} \Ff_q$ with support contained in $S$ such that $\RM[n,q,d] = \mathcal{F}_n(h)$.
\end{lemma}
\begin{proof}
   By Lemma~\ref{lm: solvable system} with probability at least  $1 - q^{-\binom{k+d^{\star}}{d^{\star}}}$ the previously described system of equations is solvable.  Suppose this is the case, and abusing notation, let $\{z_{\alpha}\}_{\alpha \in S}$ denote the solution. Define $h_1: \Ff_q^k \xrightarrow[]{} \Ff_q$ so that $h_1(\alpha) = z_{\alpha}$ if $\alpha \in S$ and $h_1(\alpha) = 0$ if $\alpha \notin S$. Then each of the previous equations states the value of an inner product involving $h_1$. In particular we have $\langle h, x^{e} \rangle = 0$ for all $e \in \{0,\ldots, q-1\}^{k}$ such that $|e|_1 \leq d^{\star}-1$ and $\langle h, x^{e^{\star}} \rangle = 0$.   By Lemma~\ref{lm: monomial inner product}, $h_1$ contains the monomial $\prod_{i=1}^{s}x_i^{q/p-1} \prod_{j =s+2}^{k} x_{j}^{q-1}$ and all other monomials in $h_1$ are degree at most $(q-1)\cdot k - (d^{\star}+1)$. We may write, 
   \[
   h_1(x) = \prod_{i=1}^{s}x_i^{q/p-1} \prod_{j =s+2}^{k} x_{j}^{q-1} + g(x),
   \]
   where $\deg(g) \leq k(q-1) - (d^{\star}+1)$.

   To complete the proof, take
    \[
    h = h_1 \cdot \left(\sum_{r'=r}^{q-1}  x_{s+1}^{q-1-r'}\right).
    \]
    From  $\prod_{i=1}^{s}x_i^{q/p-1} \prod_{j =s+2}^{k} x_{j}^{q-1}\left(\sum_{r'=r}^{q-1}  x_{s+1}^{q-1-r'}\right) $ we get each of the monomials in the second condition of Lemma~\ref{lm: sufficient local char q} and these cannot get cancelled out by $g(x) \left(\sum_{r'=r}^{q-1}  x_{s+1}^{q-1-r'}\right)$. Moreover the highest degree monomial of $h$ is
    \[
    \left(\prod_{i=1}^{s}x_i^{q/p-1}\right) x_{s+1}^{q-1-r}\prod_{j =s+2}^{k} x_{j}^{q-1},
    \]
    and its degree is precisely $k(q-1) - (d+1)$, so the first condition is satisfied as well.

    Therefore we show that with probability at least $1 - q^{-\binom{k+d^{\star}}{d^{\star}}}$ over $S$, there exists $h$ with support contained in $S$ such that $\RM[n,q,d] = \mathcal{F}_n(h)$.
    % where $\alpha_{r'} \in \Ff_q$ is uniformly random for $r \leq r' \leq q-1$. We wish for $h$ to contain each of the monomials in the second condition of Lemma~\ref{lm: sufficient local char q}. For each such monomial, its coefficient in $h$ is some nonzero linear combination of the $\alpha_{r'}'s$, hence $h$ contains the monomial with probability at least $1-1/q$. As there 
\end{proof}
We are now ready to present our basic tester, described in Algorithm~\ref{alg: basic tester q}. Henceforth, we will fix $$k = 100d^{\star}\log_q\left(\frac{100tq}{\delta} \right).$$
For $t \leq \frac{\delta}{100}q^{n/(100d^{\star}) - 1}$, we have $k \leq n$ and our basic tester is valid. The following bound on $Q(k)$ will be helpful for computations:
\[
Q(k) \leq 2q^{\frac{d+1}{q-q/p} + 1}\log(q)\left(300 \log_q\left(\frac{100tq}{\delta} \right)\right)^{d^{\star}}.
\]

\begin{algorithm}[!ht]
\caption{A Basic Tester over $\Ff_q$}
\label{alg: basic tester q}
\begin{algorithmic}[1]
\Procedure{GeneralRandomPoints}{$f, d, k$}
\State Choose $T \in \T_{n,k}$ uniformly at random.
\State Choose $S \subseteq \Ff_q^k$ of size $Q(k)$ uniformly at random.
\State Query $f \circ T(x)$ for each $x \in S$.
\State $\mathcal{H}(S) \gets $ the set of $h$ with support contained in $S$ that satisfy $\RM[n,q,d] = \mathcal{F}_n(h)$.
\If{ $\mathcal{H}(S)$ is empty}
\State \Return Accept.
\EndIf

\State Choose $h \in \mathcal{H}(S)$ uniformly at random. 
\If{ $\inner{f \circ T}{h} = 0$}
\State \Return Accept.
\EndIf
\If{ $\inner{f \circ T}{h} \neq 0$}
\State \Return Reject.
\EndIf
\EndProcedure
\end{algorithmic}
\end{algorithm}

\begin{lemma} \label{lm: basic soundness q}
    Suppose $f: \Ff_q^k \xrightarrow[]{} \Ff_q$ is $\delta$-far from $\RM[n,q,d]$. Then Algorithm~\ref{alg: basic tester q} on inputs $f,d,k$ rejects with probability at least
    $\frac{\delta}{5Q(k)^2}$.
\end{lemma}
\begin{proof}
    Suppose that $\mathcal{H}(S)$ is nonempty. Then by Theorem~\ref{th: local char to test} combined with an argument identical to Lemma~\ref{lm: random points tester}, the tester rejects with probability at least 
    $\frac{\delta}{4Q(k)^2}$.
    The probability that $\mathcal{H}(S)$ is empty is at most 
    \[
    q^{-\binom{k+d^{\star}}{d^{\star}}} \leq \frac{\delta}{q^{\binom{k+d^{\star}}{d^{\star}}-1}} \leq \frac{\delta}{q^{100d + 10d\log(k)}} \leq \frac{\delta}{100Q(k)^2}.
    \]
    Hence, we may subtract this probability out and get that $f$ is rejected with probability at least $\frac{\delta}{5Q(k)^2}$.
\end{proof}

Finally, repeating our basic tester yields a $t$-online-erasure resilient tester for $\RM[n,q,d]$.

\begin{algorithm}[!ht]
\caption{A $t$-online-erasure resilient tester over $\Ff_q$.}
\label{alg: fq test}
\begin{algorithmic}[1]
\Procedure{GeneralErasureResilient}{$f, d, \delta$}
\State Set $k = 100dt\log \left( \frac{100dt}{\delta}\right) + q$.
\For{$i = 1 \to \frac{100Q(k)^2}{\delta}$} \State Run \texttt{GeneralRandomPointsTest}$(f,d,k)$
\State If Reject is outputted, \Return Reject.
\EndFor
\State \Return Accept.
\EndProcedure
\end{algorithmic}
\end{algorithm}

\begin{theorem}
    Given $f: \Ff_q^n \xrightarrow[]{} \Ff_q$, degree parameter $d$, and a distance parameter $\delta$ in the $t$-online erasure model, Algorithm~\ref{alg: fq test} outputs accept with probability $1$ if $\deg(f)\leq d$ and outputs reject with probability $2/3$ if $f$ is $\delta$-far from degree $d$. The query complexity of Algorithm~\ref{alg: fq test} is 
    $\frac{100Q(k)^3}{\delta}$.
    % \[
    % \frac{100Q(k)^3}{\delta} \leq \frac{100}{\delta}  8q^{3\frac{d+1}{q-q/p} + 3}\log^3(q)\left(300 \log_q\left(\frac{100tq}{\delta} \right)\right)^{3d^{\star}}.
    % \]
\end{theorem}
\begin{proof}
    The query complexity is clear, and it is also clear from Lemma~\ref{lm: basic soundness q} that if  $\deg(f) \leq d$ then Algorithm~\ref{alg: fq test} outputs accept with probability $1$.
    
    To see the soundness, note that if every query is obtained successfully, i.e.\ no queried point is erased, then by Lemma~\ref{lm: basic soundness q}, reject is outputted with probability at least $3/4$ when $f$ is $\delta$-far from degree $d$. It remains to upper bound the probability that any query is erased. The total number of erased points is at most $\frac{Q(k)^3t}{\delta}$
    and each query is uniformly random in some $k$-dimensional affine subspace. Therefore by a union bound, the probability that any individual query is erased is at most,
    \[
     \frac{Q(k)^3t}{\delta q^k} \cdot \frac{Q(k)^3t}{\delta} \leq \frac{100q^{6d}\log(q)^{6}(300\log_q(100tq/\delta))^{6d^{\star}}}{\delta^2\left(100tq/\delta \right)^{100d^{\star}}} \leq \frac{1}{100}.
    \]
    Overall, this shows that if $f$ is $\delta$-far from degree $d$ then  Algorithm~\ref{alg: fq test} outputs rejects with probability at least $3/4 - 1/100 \geq 2/3$.
\end{proof}
\section{Online-Corruption Resilient Testers}
As noted, our testers also work with two sided error with the same parameters in the $t$-online-corruption model, since the probability that a corrupted point is queried at all by our algorithms is at most $1/100$. Recall that in this model, the adversary is allowed to alter entries of the input function's truth table, $f(x)$, and thus one can only hope for testers with two sided error as the probability of a querying a corrupted point must now be subtracted from both the completeness and the soundness.

A formal reduction from online-erasure resilient to online-corruption resilient testing is given in \cite[Lemma 1.8]{KRV}. Applying this lemma to our results yields the following theorems.

\begin{theorem} [Prime Field Case] \label{th: p corruption}
    Let $f: \Ff_p^n \xrightarrow[]{} \Ff_p$ be the input function over a prime-field vector space, $d$ be the degree parameter, and $\delta$ be the distance parameter. Then for $t \leq \frac{\delta}{30}p^{n/(20d)}$, there is a $t$-online-corruption resilient tester with query complexity    $O\left(\frac{\left(\log(t/\delta)\right)^{3d+3}}{\delta} \right)$
    satisfying:
    \begin{itemize}
        \item \textbf{Completeness:} If $f$ is degree $d$ then the algorithm outputs accept with probability at least $2/3$.
        \item \textbf{Soundness:} If $f$ is $\delta$-far from degree $d$ then the algorithm outputs reject with probability at least $2/3$.
    \end{itemize}
\end{theorem}

\begin{theorem} [Non-Prime Field Case] \label{th: q corruption}
    Let $f: \Ff_q^n \xrightarrow[]{} \Ff_q$ be the input function, $d$ be the degree parameter, and $\delta$ be the distance parameter. Then for $t \leq \frac{\delta}{100}q^{\frac{n}{100(d+q)} - 1}$, there is a $t$-online-corruption resilient tester with query complexity $\frac{q^{O(1)}}{\delta} O\left(\left(\log(t/\delta)\right)^{3d+3q}\right)$
    % \[
    % \frac{100}{\delta}  8q^{3\frac{d+1}{q-q/p} + 3}\log^3(q)\left(300 \log_q\left(\frac{100tq}{\delta} \right)\right)^{3d+3q}  = \frac{q^{O(1)}}{\delta} O\left(\left(\log(t/\delta)\right)^{3d+3q}\right).
    % \]
    satisfying:
    \begin{itemize}
        \item \textbf{Completeness:} If $f$ is degree $d$ then the algorithm outputs accept with probability at least $2/3$.
        \item \textbf{Soundness:} If $f$ is $\delta$-far from degree $d$ then the algorithm outputs reject with probability at least $2/3$.
    \end{itemize}
\end{theorem}

\section{Further Directions}
In this work we take a step towards investigating problems in algebraic property testing within the online erasure model. Regarding further directions, recall the gap in the $q=2$ case between our upper bound and the lower bound of  \cite{BKMR} discussed in Section~\ref{sec: related work}. Closing this gap is an intriguing avenue for future work.

Another potential direction is to develop testers for other algebraic properties or to show generally that local characterizations (as defined in Definition~\ref{def: local char}) yield local tests in the online erasure model. The latter result in the classical model is the main theorem of Kaufman and Sudan's seminal work on algebraic property testing \cite{KS}. Their main theorem is restated in this paper as Theorem~\ref{th: local char to test}.

In our current work, we require there to be many local characterizations so that any large enough set of points in a $k$-dimensional affine subspace will contain the support of some local characterization with probability close to $1$. Naively, one could try to make the same kind of argument work for a \textit{single} characterization and appeal to the Density Hales-Jewett Theorem \cite{FK, DHJ}. This would give that for a fixed local characterization $h$ and any fraction $c > 0$ there is some dimension $k_c$ such that any set of $c q^{k_c}$ points in a $k$-dimensional affine subspace will contain the support of $h$ (where $h$ is extended to the entire subspace in some manner). However, in order for this approach to work even with a constant erasure parameter, one would need $c q^{k_c} < q^{k_c/2}$, and such a statement is far stronger than what is guaranteed by the Density Hales-Jewett theorem.

Yet another intriguing direction is testing other algebraic properties that are not necessarily affine-invariant, for example, the dual BCH property which is equivalent to being the trace of a low degree polynomial from $\Ff_{q^n} \xrightarrow[]{} \Ff_q$. The set of such functions still has a 2-transitive symmetry group but this group is much less rich than the affine-invariant symmetries that we work with. 

Finally, it may be possible that our techniques can answer other questions regarding Reed-Muller codes in the online erasure model. For example, it could be interesting to explore analogues of other well studied tasks in coding theory like local decoding or local correction for Reed-Muller codes. We remark that local decoding in an (offline) erasure model appears in \cite{RRV}.
\section{Acknowledgements}
This work was done in part while the authors were visiting the Simons Institute for the Theory of Computing.
\bibliographystyle{plain}
\bibliography{references}
\end{document}